\documentclass{article}

\usepackage[preprint]{neurips_2026}

\usepackage{xspace}
\usepackage{xcolor}
\usepackage{soul}
\usepackage{mathtools}
\usepackage{microtype}
\usepackage{subcaption}
\usepackage[export]{adjustbox}
\usepackage{relsize}
\usepackage[makeroom]{cancel}
\usepackage{standalone}
\usepackage{wrapfig2}

\usepackage{tikz}
\usetikzlibrary{perspective, fit, calc, positioning, shapes, graphs, trees, patterns, 3d, shadings}
\tikzset{
  pointlabel/.style={inner sep=1pt, fill=white, above=1mm, rounded corners=1pt, font=\tiny, minimum size=7pt},
}

\usepackage{forest}
\usepackage{tikz-3dplot}

\usepackage{hyperref}
\usepackage{url}

\definecolor{darkblue}{rgb}{0, 0, 0.5}
\hypersetup{colorlinks=true, citecolor=darkblue, linkcolor=darkblue, urlcolor=darkblue}

\makeatletter
\long\def\@makefntext#1{%
  \parindent 1em%
  \noindent
  \hb@xt@1.8em{\hss\@makefnmark\hspace{0.5em}}#1%
}
\makeatother

\usepackage[capitalize, noabbrev]{cleveref}
\crefname{appendix}{Appendix}{Appendices}
\Crefname{appendix}{Appendix}{Appendices}
\usepackage{gensymb}
\usepackage{xstring}
\usepackage[outline]{contour}% http://ctan.org/pkg/contour

\usepackage[small]{complexity}
\newclass\ER{\exists\mathbb{R}}

\usepackage{amsthm}
\usepackage{thmtools}
\usepackage{thm-restate}
\theoremstyle{plain}

\newtheorem{theorem}{Theorem}
\newtheorem*{theorem*}{Theorem}

\crefname{conjecture}{Conjecture}{Conjectures}
\Crefname{Conjecture}{Conjecture}{Conjectures}

\theoremstyle{definition}
\newtheorem{definition}{Definition}

\usepackage{booktabs}
\usepackage{tabularx}
\usepackage{siunitx}
\usepackage[varbb]{newtxmath}
\usepackage{pgfplots}
\usepgfplotslibrary[fillbetween, groupplots, statistics]
\pgfplotsset{
  compat=newest,
  grid=both,
  grid style={draw=white},
  axis line style={draw=white},
  tick style={draw=none},
  axis background/.style={fill=black!10},
  legend style={draw=none, fill=black!10},
  legend cell align=left,
}

\usepackage{amsmath,bm}

\def\eqref#1{equation~\ref{#1}}
\def\1{\bm{1}}

\def\vh{{\bm{h}}}

\def\vell{{\bm{\ell}}}

\def\vw{{\bm{w}}}
\def\vx{{\bm{x}}}

\def\mA{{\bm{A}}}

\def\mD{{\bm{D}}}

\def\mI{{\bm{I}}}

\def\mQ{{\bm{Q}}}
\def\mR{{\bm{R}}}

\def\mU{{\bm{U}}}
\def\mV{{\bm{V}}}
\def\mW{{\bm{W}}}

\DeclareMathAlphabet{\mathsfit}{\encodingdefault}{\sfdefault}{m}{sl}
\SetMathAlphabet{\mathsfit}{bold}{\encodingdefault}{\sfdefault}{bx}{n}

\renewcommand{\R}{\mathbb{R}}

\newcommand{\Cov}{\mathrm{Cov}}
\DeclareMathOperator*{\argsort}{arg\,sort}

\DeclareMathOperator*{\col}{col}

\newcommand{\argsortable}[1]{\mathcal{B}\left(#1\right)}

\definecolor{clr1}{RGB}{228,26,28}
\definecolor{clr2}{RGB}{55,126,184}
\definecolor{clr3}{RGB}{77,175,74}
\definecolor{clr4}{RGB}{152,78,163}

\tikzset{
    define colours/.code n args={1}{
    
    \definecolor{clr0}{RGB}{128,128,128,128};
    \definecolor{clr1}{RGB}{228,26,28};
    \definecolor{clr2}{RGB}{55,126,184};
    \definecolor{clr3}{RGB}{77,175,74};
    \definecolor{clr4}{RGB}{152,78,163};
    \definecolor{clr5}{RGB}{255,127,0};
    \definecolor{clr6}{RGB}{247,129,191}
    \definecolor{cclr1}{RGB}{0,0,0};
    \definecolor{cclr2}{RGB}{77,175,74};
    \definecolor{cclr3}{RGB}{55,126,184};
    \definecolor{cclr4}{RGB}{66,150,129};
    \definecolor{cclr5}{RGB}{228,26,28};
    \definecolor{cclr6}{RGB}{152,100,51};
    \definecolor{cclr7}{RGB}{141,76,106};
    \definecolor{cclr8}{RGB}{120,109,95};
  },
  mycolour/.style={
         define colours={1}
  },
  color node/.style={%
      minimum height=1.5cm, minimum width=2.5cm, draw, rectangle, inner color=#1!20, outer color=#1!40, align=center, outer sep=.08cm
  },
  signvec/.style={%
      rounded corners=0.2cm, fill=white, inner sep=1pt, opacity=.75, align=center
  },
  vertex/.style={%
      circle, inner sep=0pt, outer sep=0pt, minimum size=1.8mm
  },
  set@com@col/.style={},set@com@col@aryarg/.style={column #1/.style={set@com@col}},
  set@com@row/.style={},set@com@row@aryarg/.style={row #1/.style={set@com@row}},
  set common column/.style 2 args={set@com@col/.style={#2}, set@com@col@aryarg/.list={#1}},
  set common row/.style 2 args={set@com@row/.style={#2}, set@com@row@aryarg/.list={#1}},
}

\pgfmathsetmacro{\svspacing}{0.35}

\newcommand{\ie}{i.e.\@\xspace}

\newif\ifnocomments
\nocommentstrue
\ifnocomments
  \newcommand{\andreas}[1]{}
  \newcommand{\ye}[1]{}
  \newcommand{\ag}[1]{}
  \newcommand{\mf}[1]{}
  \newcommand{\swabha}[1]{}
  \newcommand{\xiang}[1]{}
\else
  \newcommand{\andreas}[1]{%
  \par%
  \noindent\fbox{%
  \parbox{\dimexpr\linewidth-2\fboxsep-2\fboxrule}{Andreas: #1}%
  }%
  }

  \newcommand{\ye}[1]{\textcolor{purple}{(Yanai:) #1}}
  \newcommand{\ag}[1]{\textcolor{orange}{(Andreas:) #1}}
  \newcommand{\mf}[1]{\textcolor{olive}{(Matt:) #1}}
  \newcommand{\swabha}[1]{\textcolor{cyan}{(Swabha:) #1}\xspace}
  \newcommand{\xiang}[1]{\textcolor{blue}{(Xiang:) #1}}
\fi

\title{Mark My Words: Toward Unforgeable Language Model Signatures via Token Rankings}
\title{Token Rankings are Securely-Disclosable (and Maybe Unforgeable) Language Model Signatures}
\title{Let Them Have Token Rankings}
\title{Token Rankings are Unforgeable Language~Model~Signatures}
\author{
  Matthew Finlayson\thanks{Correspondence to \texttt{mfinlays@usc.edu}} 
  \\ University of Southern California
  \And Andreas Grivas \\ University of Edinburgh
  \And Xiang Ren \\ University of Southern California
  \And Swabha Swayamdipta \\ University of Southern California
}

\begin{document}

% \ifcolmsubmission
% \linenumbers
% \fi

\maketitle

\begin{abstract}
  Language model parameters are known to impose unique 
(to each model) geometric constraints on their logit outputs,
which serves as a signature that identifies the model,
but also leaks the model's final layer parameters
when an API distributes logits.
We investigate more restrictive APIs that expose token rankings 
(i.e., their ordering by probability, but not the probability values)
and find that rankings also constitute a signature: every model has a unique set of feasible top-$k$ rankings for sufficiently large $k$.
Furthermore, the ranking signature is the first known (polynomially) unforgeable signature,
since finding a model with the same set of feasible rankings
is \NP-hard.
On the security front, we find that token rankings are already sufficient to approximately steal the final layer of the model, similar to logits, though the approximation is too coarse to forge the signature, and can be effectively countered by restricting the API to top-$k$ tokens with sufficiently small $k$.
Since the top-$k$ required to present the model signature is generally smaller than the $k$ required to prevent stealing,
it is possible for an API to present an unforgeable signature without leaking model parameters.

\end{abstract}

\section{Introduction}
\label{sec:intro}

Recent work has explored using language model outputs as a signature~\citep{yang2025fingerprintlargelanguagemodels,finlayson2024logits,finlayson2026every}.
In particular, these signatures are naturally occurring, unique geometric properties of the model's outputs (e.g., token logits)
that unambiguously reveals which model generated an output.
These signatures have promising applications for model security, accountability, and forensics, 
but suffer from a few drawbacks.
Namely, the signatures generally rely on access to token probabilities, 
these token probabilities can be used to steal model parameters via an API,
and stolen parameters can in turn be used to forge the signature, 
i.e., construct a new model with the same geometric constraints.
Token probability access has become increasingly rare for language model APIs, 
presumably to protect against attacks like parameter stealing and prompt inversion~\citep{carlini2024stealing,finlayson2024logits,  ICLR2024_999fcab9,nazir2025better}.

We propose a new geometric language model signature based on token rankings, i.e., the ordering of tokens by probability, but not the probabilities themselves.
In particular, the set of token rankings that a language model produces are generally unique to that model: for every ranking the model produces, it is overwhelmingly likely that no other model can produce it.
This is due to the low-rank unembedding layer of the model---known as the softmax bottleneck~\citep{Yang2018}---which restricts the model's outputs to a proportionally tiny subset of the possible rankings.
We characterize theoretically and verify empirically the uniqueness of the ranking signature.

Furthermore, we prove that the ranking signature is unforgeable, since creating another model with the same set of rankings is computationally hard. 
We accomplish this by showing that finding a rank-$d$ matrix that can produce a set of rankings is \ER-hard where \(\NP\leq\ER\leq\PSPACE\). 
Our proof reduces a known \ER-hard problem to the ranking fitting problem.

Such a signature is useful for model forensics.
For instance, suppose you come across a repository purportedly containing leaked parameters from a private language model.
% How might you prove whether the leaked parameters are legitimate?
% If the official model has a public API
% and gives full control over the model input (i.e., has no hidden system message),
% you might check that the leaked and official models produce identical outputs on a variety of inputs, concluding that the models are not the same if they diverge. 
% Unfortunately, this test fails if the API transforms the input in some unobserved way, such as prepending a hidden system message that perturbs the behavior of the model.
% \xiang{This provides an interesting hook but I'm a bit concerned that it bias people's view on what the paper is about?}
% In this paper, we show that language models have a unique final layer signature 
If the private language model has an API that distributes token rankings, you can use the ranking signature to prove that either the leaked weights are fraudulent, or the model leaker has access to the true model weights.
% Crucially, 
% the method requires only that the API provides (top-$k$) token rankings (the ordering of tokens by probability, but not the probabilities themselves),
% and does not require knowing the input to the model API (e.g., when the API uses a secret system prompt).
% \xiang{We should speak out about the relation of model info and their rankings, and the utilities of this finding.}
%  Motivation: make clear that the signature depends only on the output layer
Concretely, it suffices to obtain a handful of token rankings $r_1,r_2,\ldots,r_n$ from the API and test whether they are feasible for the model, \ie, whether there exists $\vx_i \in \R^d$ such that $\argsort(\mW\vx_i)=r_i$ for the model's unembedding matrix $\mW$.
% and some $\vx\in\R^d$.
Checking a ranking's feasibility is efficient, 
% takes polynomial time, is a convex optimization problem that can be solved efficiently via linear programming.
and if the rankings are infeasible, then the weights are certainly not the true model weights.
On the other hand, if the rankings are feasible for the leaked model, then the leaker must have access to the true model weights because the ranking signature is both unique and unforgeable.

% Thus the feasible set of rankings---and by extension, top-\(k\) rankings---constitute a signature for the model.
% We empirically demonstrate the uniqueness of the signature:
% even the smallest perturbation to the model weights completely changes the signature.
% We then prove the hardness of forging the signature by showing that finding a rank-$d$ matrix that can produce a set of rankings is \ER-hard where \(\NP\leq\ER\leq\PSPACE\). 
% Our proof reduces a known \ER-hard problem to the ranking fitting problem.

Upon closer inspection of token rankings and the information they contain, 
we discover an important caveat for API providers who wish to provide token rankings:
APIs that expose full token rankings are susceptible to approximate final layer (\(\mW\)) stealing attacks.
% \textit{a la} \citet{carlini2024stealing,finlayson2024logits}.
We motivate the attack theoretically, and verify that it works empirically on open source models,
characterizing the sample complexity of the attack and the number \(k\) for the attack to work for top-\(k\) rankings.
Our findings overturn the prevailing assumption that unembedding stealing attacks require logprobs from the API.
While this attack is too coarse to steal the ranking signature,
it successfully recovers model parameters within a small margin.
Fortunately, there is a sweet spot for top-\(k\) rankings where \(k\) is small enough to prevent parameter stealing attacks, but large enough for the rankings to work as a signature.  

Our ranking signature is a direct improvement over prior work on geometric model signatures in terms of API assumptions, unforgeability guarantees, and protection against parameter stealing attacks.
Where previous geometric signatures require APIs to provide token probabilities, 
can be forged in polynomial time,
and inevitably leak the unembedding matrix of the model; 
the ranking signature requires only token rankings, 
is computationally hard to forge,
and enables the provider to present the signature without exposing model parameters.

\begin{figure}
  \centering
  \begin{tikzpicture}[
  font=\small,
  panel/.style={
    draw=black!22,
    rounded corners=2pt,
    fill=black!2,
    inner sep=4pt,
    minimum height=3.05cm,
    minimum width=4.35cm
  },
  title/.style={font=\footnotesize, anchor=north west},
  note/.style={font=\scriptsize, align=center},
  tiny note/.style={font=\tiny, align=center},
  token/.style={font=\tiny, inner sep=1pt, fill=white},
  box/.style={
    draw=black!45,
    rounded corners=1.5pt,
    fill=white,
    inner xsep=3pt,
    inner ysep=2pt,
    align=left,
    font=\scriptsize
  },
  arrow/.style={->, thick, black!75},
  softarrow/.style={->, semithick, black!55},
  dot/.style={circle, fill=black, inner sep=1.05pt}
]

\node[panel] (rankpanel) at (0,0) {};
\node[title] at ([xshift=2pt,yshift=-2pt]rankpanel.north west) {1. Token rankings};

\coordinate (can) at (-1.75,-0.0);
\coordinate (dont) at (-1.65,0.72);
\coordinate (will) at (-0.88,0.28);
\coordinate (must) at (+0.20,-0.32);
\coordinate (hstart) at (-2.05,-1.02);
\coordinate (hend) at (-0.05,0.68);

\draw[arrow] (hstart) -- (hend) node[above right=-1pt, font=\tiny] {$\vh$};
\foreach \p in {can,dont,will,must} {
  \draw[black!25] (\p) -- ($(hstart)!(\p)!(hend)$);
  \node[dot] at (\p) {};
}
\node[token,below=0.5mm of can] {can};
\node[token,above=0.5mm of dont] {don't};
\node[token,above=0.5mm of will] {will};
\node[token,below=0.5mm of must] {must};
\node[tiny note, align=left] at (-1.18,-1.22) {unembedding rows\\projected onto \(h\)};

\node[box, text width=1.10cm] (ranking) at (1.28,0.05) {
  \begin{tabular}{@{}rl@{}}
  1 & must\\
  2 & will\\
  3 & don't\\
  4 & can
  \end{tabular}
};
\node[tiny note] at (1.28,-1.08) {API reveals order,\\not probabilities};
% \draw[softarrow] (0.10,0.05) -- (ranking.west);

\node[panel, right=0.20cm of rankpanel] (sigpanel) {};
\node[title] at ([xshift=2pt,yshift=-2pt]sigpanel.north west) {2. Feasible sets are signatures};

\node[draw=black!25, rounded corners=2pt, minimum width=3.35cm, minimum height=1.55cm] (universe) at ([yshift=0.22cm]sigpanel.center) {};
\node[tiny note, anchor=north east] at ([xshift=-1pt,yshift=1pt]universe.north east) {all \(v!\) rankings};

\foreach \x/\y in {-1.30/0.40,-1.08/-0.30,-0.68/0.57,-0.36/-0.45,0.20/0.50,0.55/-0.42,0.88/0.08,1.22/-0.16} {
  \fill[black!18] ($(universe.center)+(\x,\y)$) circle[radius=0.55pt];
}

\draw[fill=blue!20, draw=blue!60!black, thick] ($(universe.center)+(-0.54,0.03)$) ellipse (0.55cm and 0.34cm);
\draw[fill=orange!25, draw=orange!70!black, thick] ($(universe.center)+(0.54,0.03)$) ellipse (0.55cm and 0.34cm);
\node[tiny note] at ($(universe.center)+(-0.54,0.03)$) {\(\mathcal R(\mW_1)\)};
\node[tiny note] at ($(universe.center)+(0.54,0.03)$) {\(\mathcal R(\mW_2)\)};
\node[note, text width=3.20cm] at ([yshift=-1.17cm]sigpanel.center) {nearly disjoint ranking sets identify the model};

\node[panel, right=0.20cm of sigpanel] (hardpanel) {};
\node[title] at ([xshift=2pt,yshift=-2pt]hardpanel.north west) {3. Ranking signature is unforgeable};

\node[box, font=\tiny, text width=3.05cm] (verify) at ([yshift=0.55cm]hardpanel.center) {
  Verify leaked weights: \(r\in\mathcal R(\widehat{\mW})\)? \quad Easy
};
\node[box, font=\tiny, text width=3.05cm, below=0.12cm of verify] (forge) {
  Forge signature: find \(\mW'\) with same rankings \quad \(\ER\)-hard
};
\node[box, font=\tiny, text width=3.05cm, below=0.12cm of forge] (leak) {
  Full rankings leak approximate \(\mW\); 
  top-\(k\) gives safer regime
};
\end{tikzpicture}
  \caption{
    We consider the setting where an language model API reveals token rankings, but not probabilites.
    The rankings feasible for a particular unembedding matrix form a model-specific subset of all possible rankings, yielding a geometric signature that can be checked efficiently against candidate weights.
    We prove that forging this signature is computationally hard, while also showing that full rankings can leak approximate final-layer information, motivating top-\(k\) rankings as the safer disclosure regime.
  }
\end{figure}
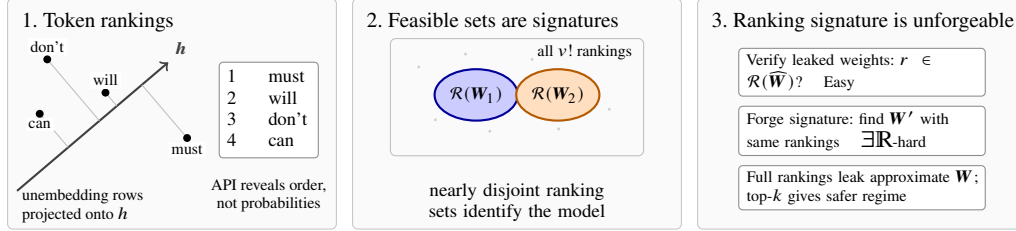

\section{Related Work}

Language model signatures are comparable to language model fingerprints and watermarks~\citep{cao2021ipguard,ye2026securing,DBLP:conf/icml/KirchenbauerGWK23} with some notable differences.
First, fingerprints and watermarks usually require the model or API provider to implement some type of modification to either the training or inference process. On the other hand, ranking signatures are naturally occurring, and all modern language models have them.
Second, fingerprints and watermarks are generally designed for robustness, i.e., trivial changes to the model weights or outputs do not remove them. In contrast, the ranking signature is extremely delicate. 
Any changes to the model weights immediately changes the signature completely.
Another difference is that the ranking signature requires access to token rankings but does not require interactive API access, while many fingerprints only need text outputs but may require interactive API access.

Similar to our unforgeability result for the ranking signature,
\citet{finlayson2026every} 
show that ellipse-based language model signatures are resistant to forgery,
since the best known algorithm for finding a language model's output ellipse takes \(\mathcal{O}(d^{2\omega})\) time\footnote{
  Here, \(2\leq \omega<3\) is the theoretical lower bound for the exponent for the time complexity~\(\mathcal{O}(n^\omega)\) of multiplying two \(n\times n\) matrices. For known practical algorithms, it is not much less than 3.
}
and requires a costly \(\mathcal{O}(d^3)\) queries from typical language model API.
This cost is high in practice, but remains polynomial.
In contrast, we show that the ranking signature is \ER-hard to forge.
Our definition of signature forgery differs from theirs,
which requires only producing a fake model output that satisfies the geometric constraint of the model parameters. 
We instead require the attacker to produce fake parameters that impose the same geometric constraints as the true parameters.
% \xiang{can we say "product rankings that all satisfy geometric constraints of the true parameters"?}
% \mf{I think no, since they explicitly need to produce parameters, not rankings}
Forging parameters is sufficient but not necessary for forging outputs.
Indeed, the authors
only show parameter forgery resistance, and only speculate about output forgery resistance. 
The forgeability of outputs remains an open question for both the ellipse and ranking signatures.
% Nevertheless, we investigate some variants of the output forgeability question,
% and prove that multi-label classifiers have unforgeable signatures under the setting where the forger must classify model outputs
% % \ag{how can multi-label classifiers classify rankings?} 
% % \mf{Nice catch. Fixed.} 
% chosen by an adversary.

\section{Token rankings are language model signatures}
\label{sec:sig}

% - What are rankings
% - What is the ranking signature
% - Ranking signatures are unique to models
% - Perturbations change the signature

We will now introduce the ranking signature.
Starting with a formal definition of feasible rankings, 
we will make a statistical argument that a model's set of feasible rankings (and top-\(k\) rankings) is unique to the model, and verify it empirically.
The uniqueness of the feasible rankings means that the set of rankings functions as a signature for the model, 
since no other model will have the same set of rankings by chance.

To formally define token rankings, consider a typical language model with embedding size~\(d\), which is less than its vocabulary size~\(v\), so it has an unembedding matrix~\(\mW\in\R^{v\times d}\) (sometimes called the language model head).
At every timestep, the model outputs a contextualized embedding $\vh$, 
which assigns a logit score~\(\vell_i\) from \(\vell=\mW\vh\) to each token~\(i=1,\ldots,v\) in the vocabulary.
Each logit vector~\(\vell\) defines a token ordering~\(r=\argsort(\vell)\), which we call a \textit{token ranking}.
We additionally define top-$k$ rankings
as a truncated ranking $r$, i.e., $(r_1,r_2,\ldots,r_k)$.
Note that the argsort of the logits is the same as the argsort of the (log-)probabilities, since the (log-)softmax function is monotonic, preserving the order. 
A~ranking~$r$ is called \textit{feasible} if there exists some \(\vh\in\R^d\) such that \(\argsort(\mW\vh)=r\), i.e., it is feasible for the model to output the ranking.\footnote{It is possible that a ranking $r$ is feasible for \(\mW\), but the model cannot (or never happens to) produce a \(\vh\) such that \(r=\argsort(\mW\vh)\). 
}
%
% \xiang{for any input to the LM, not exactly same as ``any \(\vh\)"?}
%
% \ag{Can maybe spell this out explicitly: Importantly, if there is no $\vh$ such that a given ranking is feasible, there is also no input to the LLM that can produce this ranking.}
%
While there are \(v!\) possible rankings,
a model with \(d<v\) will have a much smaller set of feasible rankings.

We call the full set of feasible rankings for a model its \textit{ranking signature}, 
so called because the set is highly unique---any single ranking is almost certainly uniquely feasible for that model and no other.
% \xiang{1. "any set of feasible" or the "entire set"? 2. "almost certainly" seems too strong?}
The uniqueness of the feasible rankings is a result of sparsity;
\citet{Cover1967} shows that \(v\times d\) matrices have at most \(\mathcal{O}(v^{2d})\) feasible rankings, 
which vanishes as a proportion of the \(v!\) possible rankings as \(v\) grows.
For intuition, \cref{fig:graph} shows the set of rankings for a small matrix as a set of nodes in a graph, where feasible rankings that differ by a single adjacent element swap are connected by edges. The sparse set of feasible rankings, which form a connected, symmetric graph, make up only about 7\% of nodes. 
\begin{restatable}{theorem}{rankingprobability}
  \label{thm:argsortable_prob_upper_bound}
  For a $v\times d$ matrix with \(d\leq v/2\), the proportion of rankings that are feasible is less than \(o(c^{-d})\) for every fixed \(c\), i.e., rankings are super-exponentially unlikely to be feasible. 
\end{restatable}
See \cref{sec:rankingcounting} for the proof.
By symmetry, each ranking has the same low probability of being feasible for a random matrix,
and so the probability that a random ranking from one random matrix is feasible for another 
is super-exponentially unlikely.
While language model heads are not random matrices,
our experiments will show that this statistical property dominates for sufficiently large \(k\). 
% \xiang{this belongs better in the next paragraph?}
% Thus, the feasible set of rankings for a language model is highly unique.

% \subsection{Top-\(k\) ranking signatures}

Top-\(k\) rankings with sufficiently large \(k\) are also model signatures.
The value of \(k\) that is sufficiently large depends on \(v\) and \(d\).
For small \(k\), e.g., \(k=1\), rankings are decidedly not unique because in high dimensions random points are likely to all be vertices of their convex hull, even when the number of points greatly exceeds the 
dimension~\citep{Donoho2006,grivas-etal-2022-low,basri2026}. As \(k\) increases we reach a sharp\footnotemark{} transition point where it becomes highly unlikely for top-\(k\) rankings to be shared.
% \cref{fig:randomtopk} illustrates this phase transition for random matrices, showing that it occurs abruptly. 
Since language model heads are not random, we might expect this phase change to occur at a slightly later point,
but we again observe in our experiments that the statistical unlikelihood of sharing top-\(k\) orderings eventually dominates. The transition occurs somewhere before the number of possible top-\(k\) rankings $_vP_k$ surpasses the number of feasible top-\(k\) rankings (bounded by the number of feasible rankings \(\mathcal{O}(v^{2d})\)).

\footnotetext{\cref{fig:randomtopk} in the Appendix demonstrates this sharpness.}

To verify the uniqueness of a model's set of feasible rankings, and to find the minimal \(k\) such that top-\(k\) rankings
% \ag{Caveat: top-k alone usually refers to the top-k set, not top-k ranking, may make sense to introduce jargon early on to make the distinction obvious.}
% \mf{Reworded for clarity. Do you have any suggestions for terminology that would make this clearer? I don't really use ``top-$k$ except to refer to top-$k$ rankings. Maybe I should just make a note of that when they first come up.}
differentiate between language models, we check whether the rankings from a target model
are feasible for other models with varying degrees of similarity to the target model 
(e.g., training data/step, model size).
We use the penultimate checkpoint for Pythia 70M trained on the deduplicated Pile dataset~\citep{pmlr-v202-biderman23a,gao2020pile}.
We obtain 100 rankings taken from a random document from the Pile, 
and find the smallest top-\(k\) that is infeasible for the other model, using a linear program to check feasibility and binary search to find the boundary.
The linear program checks the feasibility of a ranking \(r\) 
for model with head \(\mW\), 
by finding some \(\vh\) such that \(\argsort(\mW\vh)\)
has the same top-$k$.
Concretely, this means that \(\argsort(\mW\vh)\)
puts the top-\(k\) tokens from \(r\) in the same order 
(we call this the ``top-$k$ order'' constraint), 
and puts the top-$k$ tokens before all the other tokens 
(the ``top$>$bottom'' constraint).
The feasibility boundary---the minimum infeasible $k$---measures how close a ranking is to being feasible: 
the higher the boundary, the closer the full ranking is to be feasible.
\cref{fig:boundaries} presents our results.

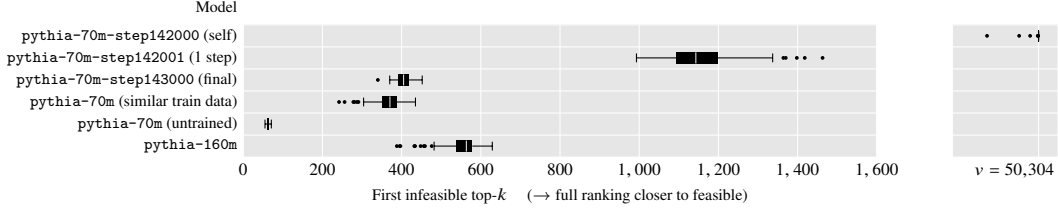
\begin{figure}
  \centering
  \tiny
\begin{tikzpicture}
  \begin{groupplot}[
      group style={
        group size=2 by 1,
        yticklabels at=edge left, 
        horizontal sep=10mm,
      },
      ytick={1,2,3,4,5,6},
      ymin=0.5, ymax=6.5,
      yticklabels={{
        \texttt{pythia-160m}
      },{
        \texttt{pythia-70m} (untrained)
      },{
        \texttt{pythia-70m} (similar train data)
      },{
        \texttt{pythia-70m-step143000} (final)
      },{
        \texttt{pythia-70m-step142001} (1 step)
      },{
        \texttt{pythia-70m-step142000} (self)
      }},
      scaled x ticks=false,
      scale only axis,
      xticklabel style={
        /pgf/number format/fixed,
        /pgf/number format/1000 sep={,},
      },
      yticklabel style={font=\tiny},
      height=1.75cm,
      ymajorgrids=false,
      minor y tick num=1,
      yminorgrids,
      width=0.6\textwidth,
      every axis plot/.append style={
        mark=*,
        mark options={fill=black, draw=none},
        mark size=0.5pt,
      },
      every boxplot/.style={
        boxplot/box extend=0.5,
        boxplot/every box/.style={
          draw=none,
          fill=black,
        },
        boxplot/every whisker/.style={
          draw=black,
        },
        boxplot/every median/.style={
          draw=black!10,
        },
      },
    ]
    \nextgroupplot[
      ylabel={Model},
      ylabel style={
        rotate=-90,
        at={(axis description cs:0,1.05)},
        anchor=south east,
      },
      xlabel={First infeasible top-\(k\) \quad (\textrightarrow{} full ranking closer to feasible)},
      xmin=0, xmax=1600,
    ]
    \addplot[
      boxplot,
      boxplot/draw position=1,
    ] table[
      y=min_infeasible_k,
      col sep=comma,
    ] {results/task12_boundaries/v5_fulltail_160m_n100_mosek.csv};
    \addplot[
      boxplot={
        draw position=2,
      },
      boxplot/draw/median/.code={
        \draw [black]
        (boxplot box cs:\pgfplotsboxplotvalue{median},0)
        --
        (boxplot box cs:\pgfplotsboxplotvalue{median},1);
      },
    ] table[
      y=min_infeasible_k,
      col sep=comma,
    ] {results/task12_boundaries/v5_fulltail_random_n100_mosek.csv};
    \addplot[
      boxplot,
      boxplot/draw position=3,
    ] table[
      y index=3,
      col sep=comma,
    ] {results/task12_boundaries/v5_fulltail_nondedup_n100_mosek.csv};
    \addplot[
      boxplot,
      boxplot/draw position=4,
    ] table[
      y=min_infeasible_k,
      col sep=comma,
    ] {results/task12_boundaries/v5_fulltail_final_checkpoint_n100_mosek.csv};
    \addplot[
      boxplot,
      boxplot/draw position=5,
    ] table[
      y=min_infeasible_k,
      col sep=comma,
    ] {results/task12_boundaries/v5_fulltail_single_weight_update_n100_mosek.csv};
    \nextgroupplot[
      width=0.1\textwidth,
      xtick={50304},
      xticklabels={$v=\num{50304}$},
      xticklabel style={anchor=north east, xshift=8pt},
      xmax={50350}, xmin=50100,
      % axis x discontinuity=crunch,
    ]
    \addplot[
      boxplot={
        draw position=6,
      },
      boxplot/draw/median/.code={
        \draw [black]
        (boxplot box cs:\pgfplotsboxplotvalue{median},0)
        --
        (boxplot box cs:\pgfplotsboxplotvalue{median},1);
      },
    ] table[
      y index=3,
      col sep=comma,
    ] {results/task12_boundaries/task12_pythia_hardness_boundaries_self_mosek.csv};
  \end{groupplot}
\end{tikzpicture}
  \caption{
    Distribution of feasibility boundaries (smallest \(k\) such that the top-\(k\) of a ranking is infeasible) 
    of \num{100} rankings from Pythia 70M checkpoint \num{142000} 
    (hidden size $d=512$).
    % \ag{What is the embedding dimensionality of the model?}
    % \mf{Detail added}
    A boundary at \(v=\num{50304}\) signifies that the ranking is feasible;
    the lower the bound, the further from feasible.
    Models are ordered by similarity to the model that produced the ranking: 
    the model itself;
    the model after a single training step;
    the model's final checkpoint; 
    a model trained on similar but different data (documents not deduplicated), 
    an untrained, randomly initialized model; 
    and a larger model ($d=768$). 
    % \ag{It might be cute to also mention the dimensionality at which it is possible to make top-k rankings feasible. I am a bit surprised by how small the top-k overlap is with the untrained model.}
    % \mf{Hm I don't quite follow, can you clarify?}
    % A few top-\(k=v\) rankings are reported as infeasible for the model itself due to failures in the solver, but boundaries remain near \(v\).
    % After a single training step, all rankings become infeasible.
    % By the final checkpoint, the rankings are just as infeasible as they are for a randomly initialized model.
    % The larger, more expressive 160M (\(d=768\)) model has a higher feasibility boundary.
  }
  \label{fig:boundaries}
\end{figure}

% TODO make explicit that we only care about the final layer

% Empirically, the similarity between the ranking signatures for any two language models of the same size depends on the models' similarity.
% The model after a single weight update is more similar to the original model than the model's final checkpoint, is more similar than a model trained on slightly different data, is more similar than a random model. 
% Other, non-derivative models of the same size should be expected to fall somewhere between the latter two similarity levels.

The signature is incredibly sensitive to perturbations.
The full rankings (i.e., \(k=v\)) are only ever feasible for the target model (though sometimes numerical precision issues cause the solver to find a slightly lower boundary).
After a single weight update to the model
(using cross entropy loss on a random document from the Pile)
with a maximum absolute change in the weights of \num{1.22e-4},
none of the tested full rankings are feasible.
% \ag{It is a bit unclear how one could prove that none of the rankings remained feasible, are these the rankings for a specific set of inputs (i.e. not the full set of rankings, right?)? Interesting. We should check whether a sample of the maximal minors also flip signs, there's a very cool function I learned about: slogdet.}
% \mf{You are right. THis is only for *tested* rankings}
Other perturbations we tried, like lowering the floating point precision of the weights (not shown), also change the signature completely.

% Models with the same vocabulary, but with larger hidden dimension $d$ have a larger set of feasible rankings, 
% % \ag{caveat: unembedding matrices need to be in general position}
% % \mf{This is true, but only for degenerate cases. How important do you think it is to be fully precise here?}
% and therefore we can expect these to have a higher feasibility boundary.
% Indeed, we see this for Pythia 160M, which has an elevated feasibility boundary. 

% TODO add caveat that we are dealing with models with unembedding layers and the same vocabulary.

% The two constraint classes (within-top-$k$ order, and top-$k$-before-bottom-$(v-k)$)
% fully explain the ordering of the feasibility boundaries between models. 
When \(k>d\), the strict top-$k$ order constraint 
makes feasibility difficult unless the unembedding matrix is almost identical to the target. 
Thus only the extremely similar 1-step model has a feasibility boundary above \(d\). 
The top$>$bottom constraints play a larger role in differentiating matrices that are not near-identical to the target.
Without it (applying only the top-$k$ order constraint) all models (except the 1-step perturbation) would have feasibility boundaries tightly clustered just barely above $d$.
With the constraint, the boundary gradually approaches $d$ from below as the model becomes more similar to the target.
For example, the final checkpoint of the same model has a higher feasibility boundary than a model trained from scratch on different data.
The boundary for Pythia 160M appears higher only because it has a larger hidden size at \(d=768\),
compared to the 70M models with \(d=512\).
For the unrelated random matrix, the feasibility boundary is extremely low, i.e., less than $k=100$.

To get a sense of the top-\(k\) required to differentiate larger models, we test a random (from the Pile) ranking from Olmo 3 8B~\citep{olmo2025olmo3} against the post-trained Olmo 3 8B Instruct, both having dimension $\num{100352}\times\num{4096}$ 
and find the boundary at \(k=\num{2481}\).
At this scale, the optimizer is slower: we are able to solve one feasibility problem at this scale every 2--3 minutes using a commercial solver~\citep{mosek}, taking about 30 minutes to run the binary search for the boundary.

\section{The ranking signature is unforgeable}
\label{sec:realization_of_rankings_is_hard}

Having shown that every language model's ranking signature is unique,
we now show that the ranking signature is \emph{unforgeable}. 
Specifically, it is computationally hard to use a set of feasible rankings from a hidden target model to construct an unembedding matrix that can produce the same set of feasible rankings.
This task, which we dub the ``signature forging problem'', is \NP-hard---in fact, it is hard for the complexity class \ER, where \(\NP\subseteq\ER\subseteq\PSPACE\)~\citep{DBLP:conf/gd/Schaefer09,DBLP:conf/stoc/Canny88}.
This is true even under the generous assumption that the attacker has access to \textit{all} the feasible rankings for a target model. 
% \ag{nitpick: maybe intractable is better word given infeasible is used for the rankings which cannot be produced.}
% \mf{True. Here and elsewhere let's use ``intractable'' for computational hardness and ``(in)feasible'' for whether a ranking can be obtained from a matrix. If you see any violations, go ahead and change them.}

We will obtain the hardness result by showing that a known \ER-hard problem is a special case of the signature forging problem.
% \xiang{We need to define ``signature forging problem".}
Specifically, and with formal definitions to follow, we will show that an algorithm that steals the signature in polynomial time can be used to decide the \textit{realizability of simple allowable sequences}~\citep{hoffmann2018universalitytheoremallowablesequences}.
The connection between these two problems (which we will soon explain) 
lies in the fact that the set of feasible rankings for a \(v\times 2\) matrix is a realizable simple allowable sequence.\footnote{
  Simple allowable sequences are sufficient for our proof,
  but can be generalized to higher dimensions $d\geq2$ in the form of sweep oriented matroids~\citep{Padrol2024},
  combinatorial abstractions of feasible ranking sets.
  }

\begin{definition}[Simple allowable sequence]
  \label{def:sas}
  A simple allowable sequence over \(v\) is a sequence of permutations of \(1,2,\ldots,v\)
  where neighbors differ by exactly one swap of adjacent elements
  and each pair of elements is reversed exactly once.
\end{definition}

For example, 
\begin{equation}
  \begin{tikzpicture}
    \node               (123) {123};
    \node[right=of 123] (132) {132};
    \node[right=of 132] (312) {312};
    \node[right=of 312] (321) {321};
    \draw (123) 
    to node[font=\tiny, above] {23} (132) 
    to node[font=\tiny, above] {13} (312)
    to node[font=\tiny, above] {12} (321)
    ;
  \end{tikzpicture}
  % (1,2,3),({1,\color{red}3,2}),({\color{red}3,1},2),(3,{\color{red}2,1})
  \label{eq:sas_ex}
\end{equation}
is a simple allowable sequence (swapped elements indicated above edges).
Notice that the last permutation in the sequence will always be a reversal of the first.

Some simple allowable sequences can be interpreted as a description of the set of feasible rankings for a \(v\times2\) matrix~\(\mW\),
\ie, the set \(\{\argsort(\mW\vx)\mid\vx\in\R^2\}\).
This becomes evident 
in \cref{fig:realizability},
where the rows of a \(3\times2\) matrix serve as coordinates of points in \(\R^2\). If we choose an arbitrary nonzero vector~\(\vx\), 
then read off the rankings obtained from \(\argsort(\mW\vx)\) as \(\vx\) is continuously rotated 180 degrees to \(-\vx\),
each intermediate ranking differs by a single swap of adjacent elements,
and each pair of elements is swapped exactly once -- a simple allowable sequence.
Reversing the intermediate rankings in the sequence gives the remaining feasible rankings for the matrix
(equivalently, we can obtain the remaining feasible rankings by continuing the rotation to a full 360 degrees).
A \(v\times2\) matrix is said to \textit{realize} the sequence obtained in this manner,
and a simple allowable sequence over \(v\) is \textit{realizable} if there is a \(v\times 2\) matrix that realizes it.

\begin{figure}
  \begin{minipage}[t]{0.47\textwidth}
    \centering
    \begin{tikzpicture}[
    scale=0.5,
    dir/.style={},
  ]
  % Arrangement
  \coordinate (one) at (1,2) {};
  \coordinate (two) at (2,1) {};
  \coordinate (thr) at (0,0) {};

  % Directions
  \coordinate (l1) at (-2,4);
  \coordinate (l2) at ($(l1) + (0,-5)$);
  \node[font=\footnotesize, below] at (l2) {$(1,2,3)$};
  \draw[-latex, dir] (l1) -- (l2);
  \draw[opacity=0.5, gray] (one) -- ($(l1)!(one)!(l2)$);
  \draw[opacity=0.5, gray] (two) -- ($(l1)!(two)!(l2)$);
  \draw[opacity=0.5, gray] (thr) -- ($(l1)!(thr)!(l2)$);
  \node[fill, inner sep=0.5pt, circle, label={[font=\footnotesize]west:1}] at ($(l1)!(one)!(l2)$) {};
  \node[fill, inner sep=0.5pt, circle, label={[font=\footnotesize]west:2}] at ($(l1)!(two)!(l2)$) {};
  \node[fill, inner sep=0.5pt, circle, label={[font=\footnotesize]west:3}] at ($(l1)!(thr)!(l2)$) {};

  \coordinate (l2) at ($(l1) + (5,1)$);
  \draw[-latex, dir] (l1) -- (l2);
  \draw[opacity=0.5, gray] (one) -- ($(l1)!(one)!(l2)$);
  \draw[opacity=0.5, gray] (two) -- ($(l1)!(two)!(l2)$);
  \draw[opacity=0.5, gray] (thr) -- ($(l1)!(thr)!(l2)$);
  \node[fill, inner sep=0.5pt, circle, label={[font=\footnotesize]north:1}] at ($(l1)!(one)!(l2)$) {};
  \node[fill, inner sep=0.5pt, circle, label={[font=\footnotesize]north:2}] at ($(l1)!(two)!(l2)$) {};
  \node[fill, inner sep=0.5pt, circle, label={[font=\footnotesize]north:3}] at ($(l1)!(thr)!(l2)$) {};
  \node[font=\footnotesize, right] at (l2) {$(3,1,2)$};

  \coordinate (l2) at ($(l1) + (4,-4)$);
  \draw[-latex, dir] (l1) -- (l2);
  \draw[opacity=0.5, gray] (one) -- ($(l1)!(one)!(l2)$);
  \draw[opacity=0.5, gray] (two) -- ($(l1)!(two)!(l2)$);
  \draw[opacity=0.5, gray] (thr) -- ($(l1)!(thr)!(l2)$);
  % \node[fill, inner sep=1pt, circle, label={[font=\footnotesize]north:1}] at ($(l1)!(one)!(l2)$) {};
  % \node[fill, inner sep=1pt, circle, label={[font=\footnotesize]north:2}] at ($(l1)!(two)!(l2)$) {};
  % \node[fill, inner sep=1pt, circle, label={[font=\footnotesize]north:3}] at ($(l1)!(thr)!(l2)$) {};
  \node[font=\footnotesize, right] at (l2) {$(1,3,2)$};

  % \node[font=\footnotesize] at (-2.25,-2) {$(1,3,2)$};
  \draw[->, shorten <=0.1mm] (l1) ++(0,-0.5) arc [start angle=-90, end angle=80, radius=0.5cm];

  \draw[-latex, dir] (l1) -- ($(l1) + (0.2,2)$) node[font=\footnotesize, above] {$(3,2,1)$};

  \node[fill, inner sep=1.5pt, circle, label=1] at (one) {};
  \node[fill, inner sep=1.5pt, circle, label=2] at (two) {};
  \node[fill, inner sep=1.5pt, circle, label=3] at (thr) {};
\end{tikzpicture}
    \caption{The realization of the simple allowable sequence~(\ref{eq:sas_ex}).
    Projecting the points onto a vector gives a token ordering. 
    Rotating the vector gives a sequence of token rankings.
    }
    \label{fig:realizability}
  \end{minipage}
  \hfill
  \begin{minipage}[t]{0.47\textwidth}
    \centering
    \includegraphics[width=0.75\textwidth,angle=90]{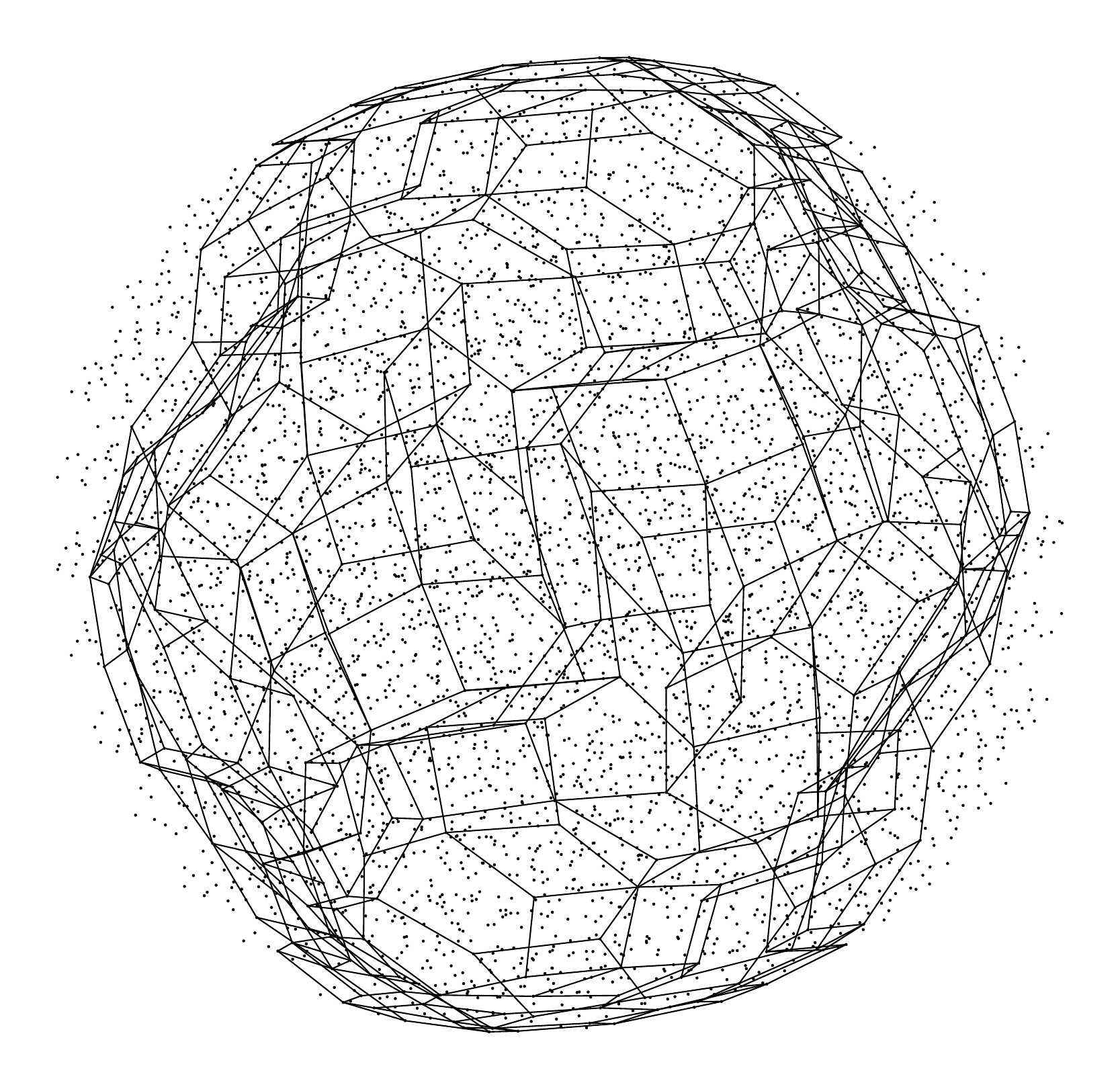}
    \caption{The feasible rankings (connected nodes) among the set of all rankings (all points) for a random $7\times3$ matrix. There are $7!=\num{5040}$ possible rankings, and \num{352} are feasible (7\%).}
    \label{fig:graph}
  \end{minipage}
\end{figure}

It may come as a surprise that not all simple allowable sequences are realizable. 
In fact, the smallest case of an unrealizable sequence appears when \(v=5\),
where
\begin{equation}
  \begin{tikzpicture}[node distance=7mm]
    \node (1)   {12345};
    \node[right=of 1]      (2)   {21345}    edge node[font=\tiny, above] {12} (1);
    \node[right=of 2]      (3)   {21435}    edge node[font=\tiny, above] {34} (2);
    \node[right=of 3]      (4)   {21453}    edge node[font=\tiny, above] {35} (3);
    \node[right=of 4]      (5)   {24153}    edge node[font=\tiny, above] {14} (4);
    \node[right=of 5]      (6)   {42153}    edge node[font=\tiny, above] {24} (5);
    \node[below=of 1]      (7)   {42513}    edge[out=north, in=south, looseness=0.1] node[font=\tiny, above] {15} (6);
    \node[right=of 7]      (8)   {42531}    edge node[font=\tiny, above] {13} (7);
    \node[right=of 8]      (9)   {45231}    edge node[font=\tiny, above]  {25} (8);
    \node[right=of 9]      (10)  {54231}    edge node[font=\tiny, above] {45} (9);
    \node[right=of 10]     (11)  {54321}    edge node[font=\tiny, above] {23} (10);
  \end{tikzpicture}
\end{equation}
is not realizable~\citep[Theorem~3.3]{Goodman1980OnTC}.
% \ag{Does this example differ from that in the Goodman/Pollack paper?}
% \mf{It is the same, except I permuted the elements so that the first and last permutations in the sequence are 12345 and 54321.}
\footnote{
  From the sequence alone, it is not obvious why this sequence is not realizable. 
  As an exercise, you can try coming up with a point arrangment that realizes it (to see that it is not possible). 
  Hint: the sequence is \textit{almost} realized by a set of points arranged in a pentagon.
}

Deciding whether a sequence is realizable is computationally hard. 
In fact, it is complete for the complexity class
known as \textit{the existential theory of reals}, 
or \ER{}, which lies between \NP{} and \PSPACE~\citep{hoffmann2018universalitytheoremallowablesequences}.
We will now show that solving the ranking signature forgery problem solves the realizability problem for simple allowable sequences.

\begin{theorem}
  Finding a \(v\times d\) matrix for which a given set of rankings are feasible is \ER-hard.
  % \ag{I think this should be full-set of rankings, \ie $\frac{|\argsortable{W}|}{2} + 1$ rankings that form an allowable sequence and not an arbitrary number of rankings? Counter-example: give me two rankings and I will give you $\mathbf{W} \in \R^{v \times 2}$ by setting its two columns to have values that realise the two needed argsorts, and x can act as a switch between the two, \ie x=[1,0] and x=[0,1]}
  % \mf{In response to Andreas, the point is not that every particular problem is hard, rather that there exist ranking sets that are hard to find a realization for. By analogy, there are many \SAT{} problems that are trivially satisfiable, but there is no known efficient algorithm for deciding \SAT{} in general.}
  % \ag{You are right, thanks for clarifying. In the wording below you say ``emits those rankings and no others'', this is partly what makes me think you need to somehow constrain the number of rankings as for a given W in general position we know the number of feasible rankings. E.g. suppose I give you a single ranking, there is no W that emits a single ranking and no others.}
  \label{thm:ranking_pattern_realizability_is_hard}
\end{theorem}

\begin{proof}[Proof]
  Suppose we have such an algorithm that takes a set of rankings, and finds in polynomial time a matrix~\(\mW\) for which those rankings are feasible. 
  Given a simple allowable sequence, we can decide whether the sequence is realizable by letting the set of rankings be the elements of the sequence and their reversals and finding a matrix~\(\mW\in\R^{v\times 2}\) with that set of feasible rankings.
  % \ag{I think allowing d to be greater than 2 here will lead to problems, \eg I could take $d=v$. Is there any reason to allow d >2?}
  % \mf{If it fits the rankings exactly, then the matrix must have rank 2. I allow d>2 in order to show that all ranking fitting algorithm must be NP-hard, even when they allow d>2.}
  % Note that if the algorithm finds a \(v\times d\) matrix with \(d>2\), then the matrix must have matrix rank 2, otherwise it would emit rankings not in the sequence---a contradiction.
  % \ag{I think this contradicts your earlier statement where you say ``emits those rankings and no others''}.
  % \mf{It is a contradiction. Made it expicit}
  % If this is the case, we can simply remove the unused dimensions via singular value decomposition to obtain a \(v\times2\) matrix.
  If such a matrix exists, and the algorithm finds it, the matrix columns constitute a point arrangement that realizes the sequence.
  Thus finding a matrix for which a set of rankings is feasible is at least as hard as simple allowable sequence realizability. Since simple allowable sequence realizability is \ER-complete, fitting a matrix to a ranking pattern must be \ER-hard.
\end{proof}
The implication of \cref{thm:ranking_pattern_realizability_is_hard} 
is that forging the ranking signature is theoretically intractable, 
since it is hard for an attacker to find a matrix with the same set of feasible rankings. 
% This result extends to top-\(k\) rankings, since top-\(k\) rankings reveal strictly less information than full rankings, and also holds in more realistic scenarios where the attacker has more restricted access to the rankings. 
% \ag{Would it make sense to give more information about what similar means, \eg we could say that from OMs we know that pinning down a realisable OM means we determine the chirotope, which restricts the matrices to those that have maximal minors with a particular sign pattern (or its negation)?}.
% \mf{Made ``similar'' precise with parenthetical}
% Our finding is good news for any inference provider who wishes to keep their model parameters secret while giving their users more granular information about the generation process.

Our analysis about the hardness of signature forgery immediately raises questions about \emph{practical} hardness.
After all, many \NP-hard problems turn out to be not so hard on average,
or may admit efficient approximation algorithms.
% One of the properties that makes \ER-complete problems hard is that their solutions may require intractably high precision, \eg, they may require writing down a super-polynomial number of bits.
% Bounding the precision of the solution 
% potentially makes the problem easier,
% perhaps even moving the problem into \NP---an interesting open question. 
Nevertheless, \ER-complete problems have historically proven difficult to approximate~\citep{DELIGKAS2022106, 9317986}, 
and \ER-hardness is considered a ``strong indicator of theoretical and practical algorithmic difficulties''~\citep{bertschinger2023training}.

To demonstrate the hardness of signature forgery,
we initialize a random \(128\times 64\) target matrix,
use it to generate a sample of rankings by randomly sampling vectors \(\vx\) and recording \(\argsort(\mW\vx)\), 
fit a new randomly initialized matrix to those rankings via gradient descent,
and test whether the target and fit matrices agree on randomly sampled rankings.
In this setup, we find that fitting more than \num{50} rankings generally fails due to the non-convexity of the problem. 
At 48 rankings, only 3/5 fitting attempts succeeded.
For each of those successful fits, we tested \num{2000} held-out rankings from the target matrix and found that none were feasible for the fit matrix.

% \paragraph{Higher dimension realizability is likely harder computationally}
% Stealing model parameters from rankings is likely hard
% because the description size of the ranking set 
% is often exponentially large.
% While the simple allowable sequences can be described by listing their \(\mathcal{O}(v^2)\) rankings,
% \citet{Cover1967} showed that for larger $d$,
% the number of feasible rankings for a general position matrix\footnotemark{} grows exponentially as~\(\mathcal{O}(v^{2d})\).
% \footnotetext{We say that a matrix $\mW \in \R^{v \times d}, v > d$, is in ``general position'' if every set of \(d\) rows is linearly independent.}
% Therefore it takes exponential time to even read the algorithm input.
% This does not preclude cleverer schemes for communicating the set of rankings, after all, the realization of the ranking set itself is a compact, implicit encoding of the set of rankings.
% Nevertheless, in a setup where the attacker is allowed to see individual rankings,
% it is likely that the attacker would need to see an exponential number of rankings before they could deduce the full set of rankings.
% We explore this idea in \cref{sec:theoryforge}, as it is pertinent to the idea of rankings as language model signatures.

\section{Rankings expose approximate model parameters (but not signatures)}

Prior work on linear and ellipse-based model signatures has required access to token probabilities (or logits) through an API,
but such detailed API access has become increasingly uncommon due to security risks, 
since token probabilities can leak model parameters---specifically the model's unembedding matrix up to a \(d\times d\) linear transformation, i.e., the column space of \(\mW\)~\citep{carlini2024stealing,finlayson2024logits}.
Does the ranking signature have the same drawback?
We show that, under certain assumptions, 
full token rankings\footnotemark{} alone
are sufficient to asymptotically recover the column space of the unembedding matrix, 
overturning the assumption that hiding token probabilities is sufficient to prevent parameter leaks.
Note that there is no contradiction between unforgeability and parameter stealing: the stealing method requires strong assumptions about the hidden state and logit distributions, and it provides no guarantee on the number of samples required to attain a high-precision approximation; forgery, by contrast, requires exact reconstruction of a highly discontinuous and sensitive combinatorial object.
% \xiang{It is confusing becoz here you said model params can be recovered from feasible rankings; (and thus you could possibly reconstruct feasible rankings at good chance?) but in Sec. 4 you claimed that rankings are unforgable.}
Fortunately, inference providers who wish to present the model signature without leaking model parameters can do so by restricting the rankings to a particular number of top-$k$ tokens so that the rankings are infeasible for other models, but the ranking information is too coarse to reconstruct the parameters.
\footnotetext{
  Exposing token rankings does not necessarily require returning them explicitly.
  For instance, an interface that allows banning specific tokens implicitly exposes the rankings, since repeatedly querying the model with the same prefix, recording the top next token and adding it to the ban list at each iteration, yields each token in its ranked order.
}

% \swabha{a little repetitive relative to the intro. what would be the utility of token rankings?}
% \mf{Tightened}
% We demonstrate that attackers cannot use token rankings to reverse engineer model parameters in polynomial time (so long as \(\P\neq\NP\)).

The intuition for our parameter stealing method is that rankings asymptotically reveal token rank correlations.
If the logits are jointly Gaussian, then the Spearman correlations of their rankings determine the Pearson correlations of the logits.
Thus, rankings recover the low-rank correlation geometry induced by the final layer.
\begin{restatable}{theorem}{colspacerecovery}
  \label{thm:colspace-recovery}
  For a language model with unembedding matrix \(\mW\) 
  whose final hidden states follow a Gaussian distribution with full-rank covariance,
  suppose all token logits have equal marginal variance, \ie,
  \(\vw_i^\top \Cov(\vh)\vw_i\) is constant across rows \(\vw_i\) of \(\mW\).
  Then the column space of \(\mW\) is recoverable from \(n\) rankings
  as \(n\to\infty\),
  in \(\mathcal{O}(nv^2+v^3)\) time.
\end{restatable}
See \cref{sec:parameter-recovery-proofs} for the proof.

% \ag{I tried to differentiate between $\rho_S$ and $\hat\rho_S$. Perhaps it would be good to move the estimation of $\hat\rho$ to the end of the proof.}

The two key assumptions for this proof are Gaussian hidden states
% which connects $\rho_S$ and $\rho$,
and equal marginal variance for token logits.
% so \(\mW\) and \(\rho\) have the same column space.
The Gaussian assumption is not meant to be a literal prediction about hidden states; rather, we use it as a simplifying assumption and first-order approximation of the true distribution.
The equal marginal variance assumption intuitively means that every token's logits vary equally, depending on the context: the token ``the'' and the token ``helix'' may have different average probabilities, but their logits vary across contexts by about the same amount.
We check each of these assumptions empirically in \cref{fig:assump} and find that the Gaussian assumption is indeed a rough approximation, but the variance assumption holds remarkably well.

\begin{figure}
  \centering
  \includegraphics[valign=t]{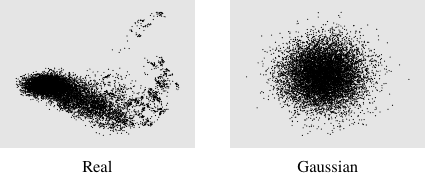}
  \hfill
  \includegraphics[valign=t]{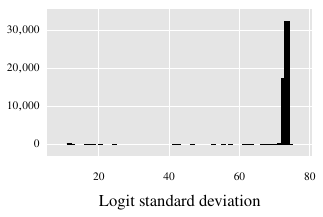}
  \caption{
    Checking assumptions for parameter stealing. Left: \num{10000} hidden states, projected onto their first two principal components.
    Real hidden states have a markedly different distribution to random Gaussian hidden states,
    making parameter stealing more difficult.
    Right: a histogram of the marginal standard deviations for each token, measured across \num{66000} real hidden states, shows a tight concentration for most tokens, with a small low-variance tail.
  }
  \label{fig:assump}
\end{figure}

The polynomial time complexity of the above parameter recovery method poses practical problems when \(v\) is large.
Fortunately, we can still approximately recover the column space when tokens have low Pearson correlations using singular value decomposition (SVD).

\begin{restatable}{theorem}{lowcorrrecovery}
  \label{thm:lowcorr-recovery}
  For a language model with unembedding matrix \(\mW\),
  whose final hidden states follow a Gaussian distribution with full-rank covariance,
  suppose all token logits have equal marginal variance and low pairwise correlation.
  Then the column space of \(\mW\) is approximately recoverable from \(n\) rankings
  as \(n\to\infty\),
  in \(\mathcal{O}(ndv)\) time.
\end{restatable}
\begin{proof}[Proof sketch]
We estimate the column space of $\mW$ from observed rankings by computing two quantities that are related by a simple formula under our assumptions.
  For each pair of tokens we get an empirical estimate of Spearman correlation, $\hat\rho_S$ using SVD, which converges to the population Spearman correlation $\rho_S$ as $n \rightarrow \infty$.
Under the assumption that the final hidden states are normally distributed, $\rho_S$ and the Pearson correlation $\rho$ are connected via a simple formula.
  When we additionally assume the tokens have low Pearson correlation, this formula is well approximated by a simple scaling factor.
  While we cannot estimate $\mW$ exactly, by making an equal-marginal-variance assumption we can estimate the column span of $\mW$ using \(\rho\).
\end{proof}
See \cref{sec:parameter-recovery-proofs} for the proof.

The proof shows that, if the Gaussian, equal-variance, and low-correlation assumptions approximately hold,
taking the singular value decomposition of a sufficient number of rankings should reveal the hidden size (via a large drop in singular values at index \(d+1\)) and approximately recover the column space of the true unembedding matrix.
Concretely, to estimate the unembedding column space from rankings,
we gather a sample of rankings~\(r_1,r_2,\ldots,r_n\in S_v\) from the model,\footnote{$S_v$ being the set of permutations over $(1,2,\ldots,v)$.}
interpret them as vectors over \(\R^v\), 
mean-center them by subtracting \((v+1)/2\),
arrange them in a \(n\times v\) matrix \(\mR\),
and employ singular value decomposition to 
the ranking matrix to obtain \(\mR=\mU\boldsymbol\Sigma\mV^\top\)
where \(\boldsymbol\Sigma\in\R^{v\times v}\) is a diagonal matrix, and \(\mU\in\R^{n\times v}\) and \(\mV\in\R^{v\times v}\) have orthonormal columns.
We observe that the diagonal entries of \(\boldsymbol\Sigma\) have a large gap in magnitude immediately following entry \(d\), revealing the hidden size of the model.
Furthermore, under the idealized assumptions above, the first $d$ columns of \(\mV\) approximately span \(\col(\mW)\), 
so \(\mV\mA\) gives an aligned estimate of \(\mW\) for some \(\mA\in\R^{d\times d}\).
% \ag{they cannot be exactly identical, otherwise the colspace of W would have been recovered, and as a result the ranking pattern would have been recovered.}
% \mf{Good catch. identical \textrightarrow close.}
% up to a \(d\times d\) non-singular linear transformation.
This is precisely the method from \citet{carlini2024stealing},
but using rankings instead of logits. 
% \ag{I feel like a cool thing to do here is to translate one model to the other and measure kl divergence of the distributions.}
% \mf{Which distributions? Probabilities obtained by using the same hidden state on both?}

\begin{figure}
  \centering
  \small
\begin{tikzpicture}[baseline=(current bounding box.center)]
  \begin{groupplot}[
    group style={
      group size=1 by 2,
      vertical sep=2mm,
      xticklabels at=edge bottom,
    },
    width=0.35\textwidth,
    height=3cm,
    xmode=log,
    log basis x=2,
    enlarge x limits,
    ylabel style={rotate=-90, font=\footnotesize, align=right, text width=3cm},
    legend style={
      at={(0.5,1.03)},
      anchor=south,
      font=\tiny,
      fill=none,
    },
    legend columns=2,
  ]
    \nextgroupplot[
      ymin=0,
      ymax=100,
      ytick={0,50,100},
      ylabel={Column space alignment (\%, \textuparrow)},
    ]
    \addplot+[mark=*, mark size=1pt] table[
      col sep=comma,
      x=n_train, 
      y expr=100*\thisrow{align_centered},
    ] {results/task13/pythia70m_spectral_centered_sweep.csv};
    \addplot+[mark=square*, mark size=1pt] table[
      col sep=comma,
      x=n_train, 
      y expr=100*\thisrow{align_centered},
    ] {results/task13/pythia70m_real_hidden_sweep.csv};
    \draw ({2^9}, 76.6) -- node[anchor=north east, font=\tiny, at end, inner sep=1pt] {Baseline} ({2^16}, 76.6);
    \legend{Gaussian, Real}
    \nextgroupplot[
      ymode=log, 
      ylabel=Relative \mbox{Frobenius} norm error (\textdownarrow),
      xlabel=Ranking samples,
    ]
    \addplot+[mark=*, mark size=1pt] table[
      col sep=comma,
      x=n_train, 
      % y expr=100*\thisrow{fro_err_centered},
      y=fro_err_centered,
    ] {results/task13/pythia70m_spectral_centered_sweep.csv};
    \addplot+[mark=square*, mark size=1pt] table[
      col sep=comma,
      x=n_train, 
      % y expr=100*\thisrow{fro_err_centered},
      y=fro_err_centered,
    ] {results/task13/pythia70m_real_hidden_sweep.csv};
    \draw ({2^9},0.621) -- ({2^16},0.621);
\end{groupplot}
\end{tikzpicture}%
%
  \sisetup{
    table-format=+1.3, round-mode=places, round-precision=3
  }%
  \hfill
  \begin{tabular}{SSSS}
    {\(\mW\)} \\
-0.1136 & -0.1075 & -0.0933 & +0.1320 \\
-0.1310 & -0.1098 & -0.1124 & +0.1509 \\
+0.0391 & -0.0226 & +0.0260 & -0.0154 \\
-0.0180 & -0.0427 & +0.0270 & -0.0103 \\
    \\
    {\(\mV\)} \\
-0.1275 & -0.1143 & -0.1077 & +0.1464 \\
-0.1308 & -0.1095 & -0.1123 & +0.1505 \\
+0.0378 & -0.0226 & +0.0265 & -0.0143 \\
-0.0192 & -0.0420 & +0.0268 & -0.0096
  \end{tabular}
  \caption{
    Approximate parameter recovery from rankings on Pythia 70M as the number of training rankings increases.
    Rankings are obtained either from random Gaussian hidden states---an idealized attack---or real hidden states sampled from Pile documents.
    For a baseline, we compute the same metrics for the unembedding matrix of the Pythia 70M model trained on non-deduplicated Pile data.
    On the right, we show the northwest \(4\times4\) corner of the true matrix \(\mW\) and recovered matrix \(\mV\) at \(2^{16}\) Gaussian samples.
    % All signs and first significant digits are correct.
  }
  \label{fig:svd-approximation}
\end{figure}
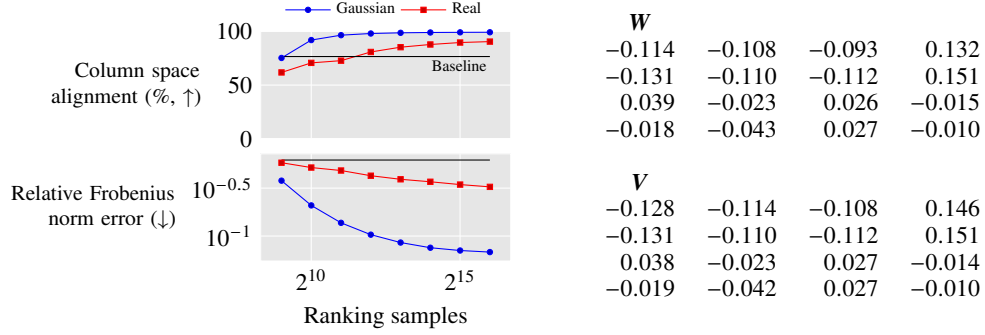

We test the extent to which we can approximately steal the final layer~\(\mW\) of \texttt{pythia-70m-dedup}.
We use synthetic Gaussian-distributed random hidden states to get an upper bound on how well the attack works,
and actual hidden states (sampled from Pile documents) for a more realistic setting.
To evaluate final-layer recovery,
we compute the cosine alignment between the principal components
\(\mathbb{E}_{i=1,\ldots,d}[\cos(\theta_i)]\)
(where \(\theta_i\) is the angle between the \(i\)th principal components of recovered and true matrices),
and the relative Frobenius norm between the true and estimated matrix, i.e., 
\(\min_{\mA\in\R^{d\times d}}{\lVert\mW-\mV\mA\rVert_F}/{\lVert\mW\rVert_F}.\)
\Cref{fig:svd-approximation} shows that the attack works fairly well, even in the realistic setting, and improves with more ranking samples. 

While this method successfully extracts the approximate unembedding matrix, it is far too coarse to forge the ranking signature. 
Empirically, the feasible rankings from our best estimate \(\mV\) are still almost entirely disjoint from the rankings of \(\mW\), and we find that the feasibility boundary for random rankings falls in the top-$k=360$--$390$ range, 
more dissimilar than Pythia 70M's penultimate and final checkpoints in \cref{fig:boundaries}.
Though the parameters from \(\mV\) asymptotically approach the true parameters under our assumptions, the asymptotic behavior of the estimation error does not appear favorable:
the asymptotic decay on the log-log error plot for the Gaussian hidden states suggests that near-zero Frobenius error would require an exponential number of samples, if it converges to zero at all.
Furthermore, any deviance from the assumptions entails irreducible error in the estimation, 
and any small error in the weights produces a completely different signature. 
From a theoretical perspective these results are inevitable,
given the \ER-hardness of signature forgery.

\begin{figure}
  \begin{tikzpicture}
  \begin{groupplot}[
      group style={
        group size=1 by 2,
        vertical sep=8mm,
        x descriptions at=edge bottom,
        y descriptions at=all,
      },
      scaled ticks=false,
      ymode=log,
      xmode=log,
      log basis x=2,
      enlarge x limits=true,
      xmin=2,
      xmax=100352,
      width=0.9\textwidth,
      xlabel={Top-\(k\)},
      ylabel={Relative \mbox{Frobenius} error},
      label style={font=\footnotesize},
      title style={font=\footnotesize, yshift=-0.75ex},
      ylabel style={
        rotate=-90, anchor=east, text width=1.5cm, align=right,
      },
      legend style={font=\tiny},
    ]
    \nextgroupplot[
      title={Pythia 70M (\texttt{pythia-70m-dedup})},
      legend pos=south east,
      height=3cm,
    ]
    % \addplot[mark=*, mark size=1pt] table[x=k, y=fro_lstsq, col sep=comma] {results/task13/topk_fro_sweep_128x.csv};
    % \addplot[mark=*, mark size=1pt] table[x=k, y=fro_lstsq, col sep=comma] {results/task13/topk_fro_sweep.csv};
    \addplot+[mark size=1pt] table[x=k, y=fro_lstsq, col sep=comma] {results/task13/topk_fro_sweep_pile_hidden.csv};
    % node [pos=0.8, above=-1pt, sloped, font=\tiny] {$2^{14}$ samples};
    \addplot+[mark size=1pt] table[x=k, y=fro_lstsq, col sep=comma] {results/task13/topk_fro_sweep_pile_hidden_128x.csv};
    % node [at end, below, sloped, font=\tiny] {$2^{16}$ samples};
    % \addplot[mark=none] coordinates {(2,0.647) (50000,0.647)}
    \addplot[mark=none] coordinates {(2,0.621) (100352,0.621)}
    node[anchor=north west, font=\tiny, at start, inner sep=1pt] {\texttt{pythia-70m} similarity};
    \addplot[mark=none] coordinates {(370.5,1)  (370.5,0.1)}
    node[anchor=east, font=\tiny, near end, inner sep=1pt] {Median feasibility boundary (\texttt{pythia-70m})} coordinate (fbend);
    % \addplot[mark=none] coordinates {(435,1)  (435,0.05)}
    % node[anchor=east, font=\tiny, near end, inner sep=1pt, fill=black!10] {Max bound};
    \addplot[mark=none] coordinates {(512,1)  (512,0.1)}
    node[anchor=west, font=\tiny, near end, inner sep=1pt] {Hidden size} coordinate (hsend);
    % \draw[<->, shorten <=1pt, shorten >=1pt] (fbend) -- (hsend) node[right, font=\tiny] {} ;
    \legend{$2^{14}$ samples, $2^{16}$ samples}
    \nextgroupplot[
      title={OLMo 3 8B},
      legend pos=south west,
      height=3cm,
    ]
    \addplot+[mark size=1pt] table[x=k, y=fro_lstsq, col sep=comma] {results/task13/olmo_topk_fro_sweep.csv};
    \addplot[mark=none] coordinates {(2481,1)  (2481,0.1)}
    node[anchor=east, font=\tiny, midway, inner sep=1pt] {Feasibility boundary (\texttt{olmo-3-8b-instruct})};
    \addplot[mark=none] coordinates {(4096,1)  (4096,0.1)}
    node[anchor=west, font=\tiny, near end, inner sep=1pt] {Hidden size};
    \addplot+[mark=none] coordinates {(2,0.6215609425371542) (100352,0.6215609425371542)}
    node[anchor=north west, font=\tiny, at start, inner sep=1pt] {\texttt{olmo-3-8b-instruct} similarity};
    \legend{$2^{15}$ Gaussian samples}
  \end{groupplot}
\end{tikzpicture}
  \caption{
    Exposing the signature, but not the parameters. 
    Top: as top-\(k\) increases, the recovered matrix error for \texttt{pythia-70m-deduped} decreases.
    However, it only surpasses the similarity from a different model (\texttt{pythia-70m}) well after the value $k=370$ where 50/100 rankings are infeasible for that model.
    Thus, the API can expose the signature without exposing parameters roughly when \(370<k<\num{16000}\).
    Bottom: the same top-\(k\) attack on \texttt{olmo-3-8b} shows a similar tradeoff, with the recovered matrix remaining less similar to the target than the model-level baseline at the conservative setting \(k=d=4096\).
  }
  \label{fig:cross}
\end{figure}
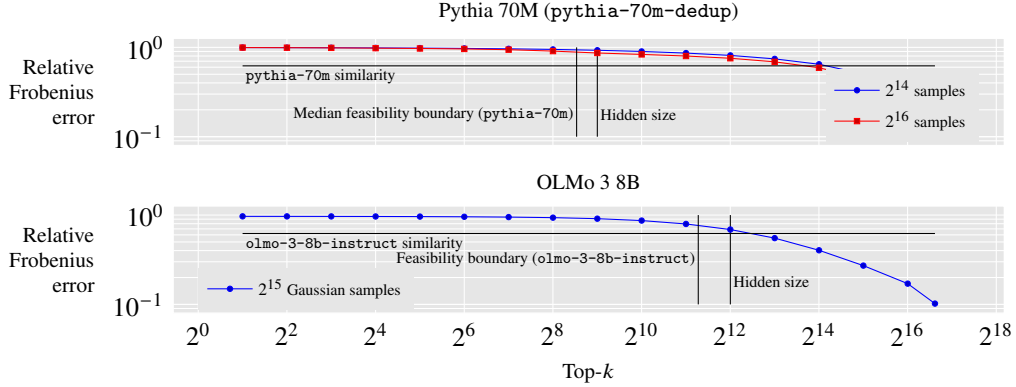

Token rankings appear to be a double-edged sword:
they are both an unforgeable model signature,
and a security risk that could leak information about the model size and parameters. 
However, it may be possible to present the model signature without exposing parameters. 
In particular, \cref{fig:cross} shows that there is a range of values of \(k\) that is just the right size 
so that an API that gives only top-$k$ rankings 
presents the signature of \texttt{pythia-70m-dedup} 
but conceals its unembedding matrix.\footnote{
  To approximate the unembedding matrix from top-$k$ tokens,
  we apply the same method as for rankings,
  except the bottom tokens all get the same rank.
}
In particular, a conservative setting of \(k=d\) is sufficient to
differentiate the signature of \texttt{pythia-70m},
while the recovered parameters are less similar to the target model than the \texttt{pythia-70m} unembedding matrix.
The same qualitative tradeoff appears for \texttt{olmo-3-8b} in the second panel of \cref{fig:cross}. 
From an intuitive perspective, if all the top-$k$ rankings are feasible for \texttt{pythia-70m}, then its unembedding matrix cannot be ruled out as a candidate for the true matrix,
so it would make sense that $k$ would need to be beyond the feasibility boundary in order to accurately recover the true matrix.
Since only extremely similar models have feasibility boundaries beyond \(k=d\), an inference provider could reasonably choose $k=d$ as a limit for top-$k$ tokens, making $d$ known, but protecting the model weights and presenting the signature.

\section{Conclusion}

We identify model rankings as a new type of model signature---the first that does not require access to token probabilities.
Additionally, we prove theoretically that the ranking signature is unforgeable.
We also uncover a new parameter stealing attack for APIs that distribute token rankings, but show that the attack can be effectively countered by restricting outputs to top-$k$ rankings while preserving the ranking signature.
Our findings have practical implications for language model providers, and provide a new tool for model forensics.

Future work might improve upon the ranking signature by finding a geometric signature with stronger unforgeability guarantees.
Though \ER-hardness is a formidable obstacle, it is not known to be cryptographically secure.
The average case hardness of \ER{} problems is not currently known as far as we are aware.
As another direction, 
the ranking signature may be applicable for other types of models. 
While we focus exclusively on language models in this work,
all that is required for a ranking signature is a rank bottleneck at the output layer to restrict the set of feasible rankings. 
For instance, multilabel classifiers may have a similar unforgeable signature, since a rank bottleneck in the output sigmoid layer restricts the set of label combinations that the model can output, which has connections to the \ER-hard problem of oriented matroid realizability~\citep{Ziegler1995}.

% Another limitation of the ranking signature is that it is tied exclusively to the unembedding matrix of the model.
% Preserving the signature while changing the model is possible if the unembedding matrix stays constant.
% Thus, if the API rankings match the leaked model's signature in the scenario from \cref{sec:intro}, we can conclude that the leaker has access to the true model weights, but we still do not know whether the leaked weights are the same as the API model's.
% Future work might improve this method by finding geometric constraints that originate from different parts of the model,
% enabling in order to verify the whole model.

% \citet{lienkaemper2023} studied the rank of a matrix to which we apply a monotonic function to the whole matrix (\textit{underlying rank}), or a different monotonic function to each row of the matrix (\textit{monotone rank}).
% For the proof, see \citet[Theorems 4,5,6]{lienkaemper2023}, which follow a similar rationale to the one provided here.

% In the next section, we will argue in favor of an conjecture that is even stronger---that an attacker cannot feasibly deduce a description of the set of rankings from the model interface when \(D\) is sufficiently large. 

\section{Acknowledgements}

Matthew Finlayson is supported by a fellowship from the National Science Foundation (NSF) Graduate Research Fellowship Program. 
Andreas Grivas is supported by ERC grant ``Numerical Analysis for Stable AI'', 101198795.
Swabha Swayamdipta’s research is supported in part by the NSF under grant IIS2403437, the Simons Foundation, Apple, Intel and the Allen Institute for AI. Part of this research was done when Swabha Swayamdipta and Matthew Finlayson were visitors at the Simon’s Institute. Xiang Ren’s research is supported in part by the Office of the Director of National Intelligence (ODNI), Intelligence Advanced Research Projects Activity (IARPA), via the HIATUS Program contract \#2022-22072200006, the Defense Advanced Research Projects Agency with award HR00112220046, and NSF IIS 2048211. We would like to thank all the collaborators in INK and DILL research labs at USC and at the Simons Institute at UC Berkeley for their constructive feedback on the work, and Yanai Elazar for the question that originally sparked the project.

% \nocite{*}
% \printbibliography
\bibliography{main}
\bibliographystyle{colm2026_conference}

\appendix 
\crefalias{section}{appendix}

\section{Random rankings are super-exponentially unlikely to be feasible}
\label{sec:rankingcounting}

When a matrix~$\mW \in \R^{v \times d}$ is in general position (\ie any submatrix formed by $d$ rows have non-zero determinant), 
the number of feasible rankings is known and depends only on $v$ and $d$.
\begin{theorem*}[\citealp{Cover1967}]
  Consider the set of feasible rankings, $\argsortable{\mW} = \left\{\argsort \mW\vx: \vx \in \R^{d}\right\},\, \mW \in \R^{v \times d}$, where the rows of $\mW$ are in general position.
  We have:
  \begin{equation}
      \left| \argsortable{\mW} \right| = 2 \sum_{d'=0}^{d-1} S_{v, d'}
  \end{equation}
  Where:
  \begin{equation}
    S_{v, d} = \sum_{A\subseteq\{2,3,\ldots,v-1\}, |A|=d}\prod_{i \in A} i
  \end{equation}
  i.e., $S_{v,d}$ is the sum of all products of $d$-element subsets of $\{2,3,\ldots,v-1\}$. 
\end{theorem*}
  See~\citet{Good1977} for the affine case.

% \ag{I think we need to define $S_{v, 0} = 1$? We know that $Q(v, 2) = 2$, which by the formula is $2 S_{v, 0}$.}
% \mf{I think this follows from the definition: there is only 1 0-element subset of $\{2,3,\ldots,v-1\}$, and the empty product is vacuously equal to 1.}

The number of rankings an LM parametrised by $\mW$ has~\(Q(v,d)\) can also be expressed as a recurrence relation \(Q(v, d) = Q(v-1,d) + (v-1)Q(v-1,d-1)\).
This grows astronomically large very quickly, with \(Q(32000,400)=\num{4.4e2608}\).
However, this is miniscule compared to the number of \emph{possible} rankings which is \(v!\) which is on the order of \num{e130271}, obtained using Stirling's approximation.
\Cref{fig:overlap_prob} shows how the probability of an LM having any particular ranking (\(Q(v,d)/v!\)) drops quickly to near 0 as soon as $d<v$.
In fact, the probability of having a particular ranking shrinks super-exponentially as the size of the LM increases (when \(v\geq 2d\), as is generally the case for language models), as we will now show in \cref{thm:argsortable_prob_upper_bound}.

% \ag{Maybe worth explicitly motivating the choice of $v > 2d$, which presumably comes from characteristics of unembedding matrices in subword-level LLMs?}
% \mf{Done}

\begin{figure}
  \centering
  \begin{tikzpicture}[font=\small]
    \begin{axis}[
        colormap/redyellow,
        view={-20}{30},
        xlabel=\(v\), ylabel=\(d\), zlabel=Probability,
        zmin=0, zmax=1,
        ytick={10, 20},
        height=4cm,
        width=\textwidth,
      ]
      \addplot3 [
        surf,
        mesh/ordering=y varies,
        shader=faceted,
        opacity=0.5
      ] table {data/overlap_prob.dat};
    \end{axis}
  \end{tikzpicture}
  \caption{
    The probability that a randomly selected ranking can be output by a language model with hidden size \(d\) and vocab size \(v\).
  }
  \label{fig:overlap_prob}
\end{figure}
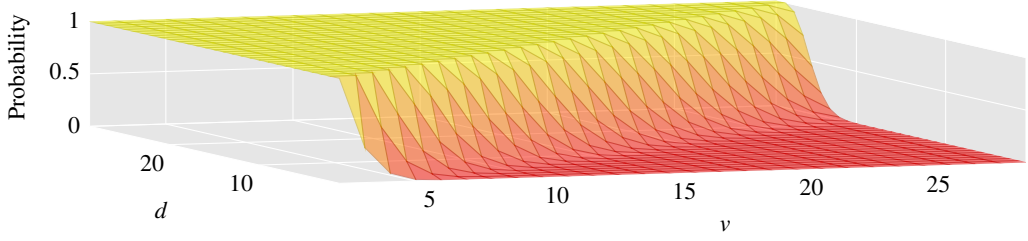

\rankingprobability*

\begin{proof}
  The probability of a random permutation (chosen uniformly at random) being argsortable by a general position matrix is equal to the number of argsorts the matrix has divided by the total number of permutations. Since we consider $v \geq 2d$,
  this probability is greatest when \(v=2d\), as \(Q\) is monotonic in \(d\).
Thus we have
\begin{align}
  \frac{|\mathcal{B}(\mW)|}{(2d)!} &= \frac{2}{(2d)!}\sum_{d=0}^{d-1}\sum_{A\subseteq\{2,3,\ldots,2d-1\},|A|=d}\prod_{i\in A}i \\
  &< \frac{2}{(2d)!}\sum_{d=0}^{d-1}\binom{2d-2}{d}\frac{(2d-1)!}{(2d-1-d)!} \\
  &= \frac{2}{(2d)!}\sum_{d=0}^{d-1}\frac{(2d-2)!}{d!(2d-2-d)!}\frac{(2d-1)!}{(2d-1-d)!} \\
  &< \frac{2d(2d-2)!(2d-1)!}{(2d)!(d-1)!(2d-1-d)!(2d-d)!} & \text{When \(v>2d\)}\\
  &= \frac{d^2(2d-2)!}{d!^3} < \frac{d^2(2d)!}{d!^3} \\
  &< \frac{d^2e(2d)^{2d+1/2}e^{-2d}}{\left(\sqrt{2\pi}d^{d+1/2}e^{-d}\right)^3} & \text{Stirling's} \\
  &= \frac{de^{d+1}2^{2d}}{(2\pi)^{3/2}d^d} < \frac{32^{d}}{d^d}
\end{align}
which means that the probability shrinks super-exponentially as \(d\) (and therefore \(v\)) grows.
\end{proof}

\section{Proofs for parameter recovery from rankings}
\label{sec:parameter-recovery-proofs}

\colspacerecovery*

\begin{proof}
  Let $\mR=[r_1,r_2,\ldots,r_n]^\top \in \R^{n \times v}$ be a matrix of mean-centered rankings from 
  the model. 
  Singular value decomposition gives us \(\mR=\mU\Sigma\mV^\top\) in \(\mathcal{O}(nv^2)\) time. 
  From that, we can get an empirical estimate of the Spearman correlation \(\hat\rho_S=\mD_S^{-1}\mV\Sigma^2\mV^\top\mD_S^{-1}\) with entries indexed by pairs of tokens, where \(\mD_S\) is a diagonal matrix that normalizes the diagonal entries of \(\hat\rho_S\) to be 1.
  This multiplication to obtain \(\hat\rho_S\) can be done in \(\mathcal{O}(v^3)\) time.
  Because the hidden states are normally distributed, the logits are too, since they are a linear transformation of the hidden states.
  We can therefore use the classical result~\citep[Equation 6.4]{kruskal1958}
  % \ag{Is there a better ref?}
  % \mf{IDK man, Kruskal doesn't give a reference, says it's "easy to derive" lol.}
  that the Pearson correlation for the logits \(\rho\) is \(2\sin(\pi\rho_S/6)\) (applied element-wise).
  % \ag{Adjust order of steps and move following line to end of proof?} 
  % \mf{looks fine to me here}
  We use the estimator $\hat\rho \coloneqq 2\sin(\pi\hat\rho_S/6)$ which converges to $\rho$ as \(n\to\infty\).
  By definition, the Pearson correlation matrix is \(\rho=\mD^{-1}\mW\Cov(\vh)\mW^T\mD^{-1}\) 
  where \(\Cov(\vh)\) is the covariance matrix for the hidden states,
  and \(\mD\) is diagonal, chosen to normalize the diagonal of \(\rho\) to 1.
  Therefore, \(\rho\) has rank \(d\),
  and the eigendecomposition \(\hat\rho=\mQ\Lambda\mQ^\top\)---obtained 
  in \(\mathcal{O}(v^3)\) time---will have \(d\) non-zero eigenvalues as \(n\to\infty\).
  By the equal-marginal-variance assumption, \(\mD=c\mI\) for some scalar \(c>0\).
  The corresponding \(d\) columns of \(\mQ\) therefore span \(\col(\mD^{-1}\mW)=\col(\mW)\).
\end{proof}

\lowcorrrecovery*

\begin{proof}
  Let $\mR=[r_1,r_2,\ldots,r_n]^\top \in \R^{n \times v}$ be a matrix of mean-centered rankings from 
  the model. 
  SVD gives us \(\mR=\mU\Sigma\mV^\top\) 
  for the first \(d\) singular values
  in \(\mathcal{O}(ndv)\) time. 
  We will show that the first $d$ columns of \(\mV\) approximately span \(\col(\mW)\) as \(n\to\infty\).
  As in the proof of~\cref{thm:colspace-recovery}, we know that \(\hat\rho_S=\mD_S^{-1}\mV\Sigma^2\mV^\top\mD_S^{-1}\) for a diagonal-normalizing \(\mD_S\), and \(\rho_S=\frac{6}{\pi}\arcsin(\rho/2)\).
  When the Pearson correlations \(\rho\) are small,
  \(\arcsin(\rho/2)\approx\rho/2\),
  so we have \(\rho_S\approx\frac{3}{\pi}\rho=\frac{3}{\pi}\mD^{-1}\mW\Cov(\vh)\mW^\top\mD^{-1}\).
  By the equal-marginal-variance assumption, \(\mD=c\mI\), so \(\col(\rho)=\col(\mW)\).
  Thus, the top \(d\) right singular vectors of the ranking matrix approximately span \(\col(\mW)\).
\end{proof}

\begin{figure}
  \centering
    \begin{tikzpicture}[font=\small]
    \begin{axis}[
        ymin=0, ymax=1,
        xmin=2, xmax=100,
        ytick={0,0.5,1},
        ymajorgrids=false,
        width=0.9\textwidth, height=3cm,
        clip=false,
        xlabel={Top-\(k\)},
        ylabel=Probability,
        ylabel style={rotate=-90}
      ]
      \addplot[mark=*, mark size=1pt, thick] table {data/phase_change.dat};
      \addplot[draw=none, name path=lower] table[y index=3] {data/phase_change.dat};
      \addplot[draw=none, name path=upper] table[y index=4] {data/phase_change.dat};
      \addplot[opacity=0.1] fill between [of=lower and upper];
    \end{axis}
  \end{tikzpicture}
  \caption{
    Probability that a random top-\(k\) ranking is feasible for random \(\num{10 000}\times\num{100}\) matrices sampled from the standard normal distribution~\(\mathcal{N}(0,1)\). 
    The shaded region indicates the 95\% confidence interval, estimated from 30 sampled matrices for each top-\(k\) size.
    The data show a phase-change between \(k=9\) and \(k=16\), where the matrices go from high to low coverage of top-\(k\) argsorts.
  }
  \label{fig:randomtopk}
\end{figure}
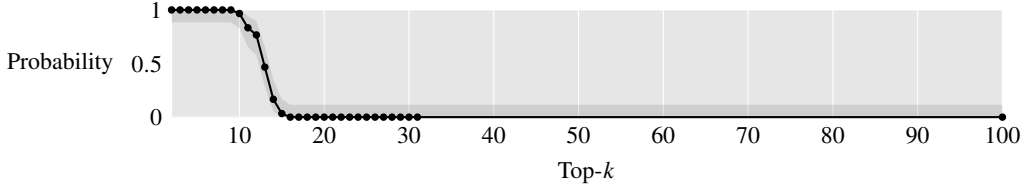

\end{document}